\documentclass[a4paper,12pt]{elsarticle}
\usepackage[T1]{fontenc}
\usepackage[english]{babel}
\usepackage{ae,aecompl}
\usepackage{amsmath, amsthm}
\usepackage{amssymb}
\usepackage[utf8]{inputenc}
\usepackage[section] {placeins}
\usepackage[ruled, vlined, linesnumbered]{algorithm2e}
\newtheorem {Theorem}                 {Theorem}         [section]
\newtheorem {theorem}      [Theorem]  {Theorem}
  
\newtheorem {myAlgorithm}    [Theorem]  {Algorithm}

\newtheorem {lemma}        [Theorem]  {Lemma}

\input amssym.def
\input amssym

\usepackage{pgfplots}
\usepackage{verbatim}
\usepackage{subfigure}
\usepackage{tikz}
\usetikzlibrary{shapes,arrows,backgrounds,mindmap}
\usepackage{titlesec}
\journal{arXiv}

\begin{document}
	\begin{frontmatter}
		\title{Computing $2$-twinless blocks}
		\author{Raed Jaberi}
		
		\begin{abstract}  
		Let $G=(V,E))$ be a directed graph. A $2$-twinless block in $G$ is a maximal subset $B\subseteq V$ of size at least $2$ such that for each pair of distinct vertices $x,y \in B$, and for each vertex $w\in V\setminus\left\lbrace x,y \right\rbrace $, the vertices $x,y$ are in the same twinless strongly connected component of $G\setminus\left \lbrace w \right\rbrace $.
		In this paper we present algorithms for computing the $2$-twinless blocks of a directed graph.
		\end{abstract} 
	\end{frontmatter}
	\section{Introduction}
	Let $G=(V,E)$ be a directed graph. $G$ is \textit{twinless strongly connected} if it contains a strongly connected spanning subgraph $(V,E^t)$ such that $E^t$ does not contain any pair of antiparallel edges. A \textit{twinless strongly connected component} of $G$ is a maximal subset $C_t \subseteq V$ such that the induced subgraph on $C_t$ is twinless strongly connected. A \textit{strong articulation point} of $G$ is a vertex whose removal increases the number of strongly connected components of $G$. A \textit{strong bridge} of $G$ is an edge whose deletion increases the number of strongly connected components of $G$. A strongly connected graph is $2$-vertex-connected if it has at least $3$ vertices and it has no strong articulation points. A $2$-vertex-connected component of $G$ is a maximal vertex subset $C^v\subseteq V$ such that the induced subgraph on $C^v$ is $2$-vertex-connected.
	A \textit{$2$-directed block} in $G$ is a maximal vertex subset $B^d\subseteq V$ with $|B^d|>1$ such that for any distinct vertices $x,y \in B^d$, the graph $G$ contains two vertex-disjoint paths from $x$ to $y$ and two vertex-disjoint paths from $y$ to $x$. 	A \textit{$2$-edge block} in $G$ is a maximal subset $B^{eb}\subseteq V$ with $|B^{eb}|>1$ such that for any distinct vertices $v, w \in B^{eb}$, there are two edge-disjoint paths from $v$ to $w$ and two edge-disjoint paths from $w$ to $v$ in $G$. A \textit{$2$-strong block} in $G$ is a maximal vertex subset $B^{s}\subseteq V$ with $|B^{s}|>1$ such that for each pair of distinct vertices $x,y \in B^{s}$ and for every vertex $u \in V\setminus \lbrace x,y\rbrace$, the vertices $x$ and $y$ are in the same strongly connected component of the graph $G\setminus \lbrace u\rbrace$. 
	A \textit{twinless articulation point} of $G$ is a vertex whose removal increases the number of twinless strongly connected components of $G$.  A $2$-twinless block in $G$ is a maximal vertex set $B\subseteq V$ of size at least $2$ such that for each pair of distinct vertices $x,y \in B$, and for each vertex $w\in V\setminus\left\lbrace x,y \right\rbrace $, the vertices $x,y$ are in the same twinless strongly connected component of $G\setminus\left \lbrace w \right\rbrace $. Notice that $2$-strong blocks are not necessarily $2$-twinless blocks (see Figure \ref{fig:2twinlessblocksexample}). 

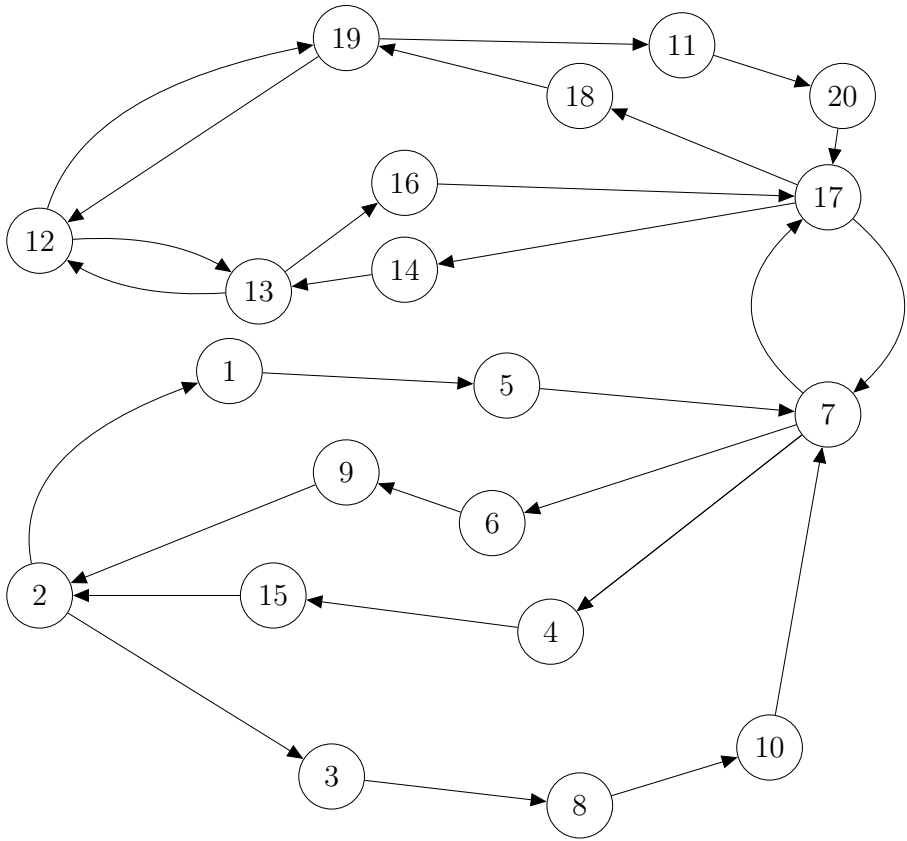
\begin{figure}[h]
	\centering
	\scalebox{0.96}{
		\begin{tikzpicture}[xscale=2]
		\tikzstyle{every node}=[color=black,draw,circle,minimum size=0.9cm]
		\node (v1) at (-1.2,3.1) {$1$};
		\node (v2) at (-2.5,0) {$2$};
		\node (v3) at (-0.5, -2.5) {$3$};
		\node (v4) at (1,-0.5) {$4$};
		\node (v5) at (0.7,2.9) {$5$};
		\node (v6) at (0.6,1) {$6$};
		\node (v7) at (2.9,2.5) {$7$};
		\node (v8) at (1.2,-2.9) {$8$};
		\node (v9) at (-0.4,1.7) {$9$};
		\node (v10) at (2.5,-2.1) {$10$};
	
		\node (v11) at (1.9,7.6) {$11$};
		\node (v12) at (-2.5,4.9) {$12$};
		\node (v13) at (-1, 4.2) {$13$};
		\node (v14) at (0,4.5) {$14$};
		\node (v15) at (-0.9,0) {$15$};
		\node (v16) at (0,5.7) {$16$};
		\node (v17) at (2.9,5.5) {$17$};
		\node (v18) at (1.2,6.9) {$18$};
		\node (v19) at (-0.4,7.7) {$19$};
		\node (v20) at (3,6.9) {$20$};

		\begin{scope}   
		\tikzstyle{every node}=[auto=right]   
		\draw [-triangle 45] (v1) to (v5);
		\draw [-triangle 45] (v2) to (v3);
		\draw [-triangle 45] (v2) to [bend left ](v1);
		\draw [-triangle 45] (v7) to (v4);
		\draw [-triangle 45] (v5) to (v7);
		
		\draw [-triangle 45] (v7) to (v6);
	\draw [-triangle 45] (v9) to (v2);
	\draw [-triangle 45] (v7) to (v4);
	\draw [-triangle 45] (v4) to (v15);
	\draw [-triangle 45] (v15) to (v2);
	\draw [-triangle 45] (v10) to (v7);
	\draw [-triangle 45] (v3) to (v8);
	\draw [-triangle 45] (v6) to (v9);
	\draw [-triangle 45] (v8) to (v10);
	\draw [-triangle 45] (v7) to [bend left ](v17);
	\draw [-triangle 45] (v17) to [bend left ](v7);
	\draw [-triangle 45] (v17) to (v18);
	\draw [-triangle 45] (v20) to (v17);
	\draw [-triangle 45] (v19) to (v12);
	\draw [-triangle 45] (v12) to [bend left ](v19);
	\draw [-triangle 45] (v12) to [bend left ](v13);
	\draw [-triangle 45] (v13) to [bend left ](v12);
	\draw [-triangle 45] (v13) to (v16);
	\draw [-triangle 45] (v16) to (v17);
	\draw [-triangle 45] (v17) to (v14);
	\draw [-triangle 45] (v14) to (v13);
	\draw [-triangle 45] (v18) to (v19);
	\draw [-triangle 45] (v19) to (v11);
	\draw [-triangle 45] (v11) to (v20);
		\end{scope}
		\end{tikzpicture}}
	\caption{A strongly connected graph $G$, which contains two $2$-strong blocks $C_{1}=\left\lbrace 2,7\right\rbrace , C_{2}=\left\lbrace 12,13,17,19\right\rbrace $, and one $2$-twinless block $B=\left\lbrace 2,7 \right\rbrace $. Notice that the vertices $12$ and $17$ do not belong to the same twinless strongly connected component of $G\setminus\left\lbrace 13 \right\rbrace $.}
	\label{fig:2twinlessblocksexample}
\end{figure}
A twinless strongly connected graph $G$ is said to be $2$-vertex-twinless-connected if it has at least three vertices and it does not contain any twinless articulation point \cite{Jaberi19}. A $2$-vertex-twinless-connected component is a maximal subset $U^{2vt} \subseteq V$ such that the induced subgraph on $U^{2vt}$ is $2$-vertex-twinless-connected. While $2$-vertex-twinless-connected components have at least linear number of edges, the subgraphs induced by $2$-twinless blocks do not necessarily contain edges.
	
Strongly connected components can be found in linear time \cite{T72}. In $2006$, Raghavan \cite{SR06} showed that the twinless strongly connected component of a directed graph can be found in linear time.
In $2010$, Georgiadis \cite{G10} presented an algorithm to check whether a strongly connected graph is $2$-vertex-connected in linear time. Italiano et al. \cite{ILS12} gave linear time algorithms for identifying all the strong articulation points and strong bridges of a directed graph. Their algorithms are based on dominators \cite{FILOS12,FILOS16,AHLT99,BGKRTW00,LTW2006,LT79}. In $2014$, Jaberi \cite{J14} presented algorithms for computing the $2$-vertex-connected components of directed graphs in $O(nm)$ time (published in \cite{Jaberi16}). Henzinger et al.\cite{HKL15}
gave algorithms for calculating the $2$-vertex-connected components in $O(n^{2})$ time. Jaberi \cite{Jaberi15} presented algorithms for computing $2$-blocks in directed graphs.  Georgiadis et al. \cite{GILP14SODA,LILPACM2016} gave linear time algorithms for determining $2$-edge blocks. Georgiadis et al. \cite{GILP14VertexConnectivity,GILP2018} also gave linear time algorithms for calculating $2$-directed blocks and $2$-strong blocks. 
Georgiadis et al. \cite{GIKPALENEX2018} and Luigi et al. \cite{LGILP15} performed experimental studies of recent algorithms that calculate $2$-blocks and $2$-connected components in directed graphs.
In $2019$, Jaberi \cite{Jaberi19} presented an algorithm for computing $2$-vertex-twinless-connected components.  Georgiadis and Kosinas \cite{GeorgiadisKosinas20} gave a linear time algorithm for calculating twinless articulation points.

In the following section  we show that  the $2$-twinless blocks of a directed graph can be calculated in $O(n^{3})$ time.
\section{Algorithm for computing $2$-twinless blocks}
 In this section we present an algorithm for computing the $2$-twinless blocks of a twinless strongly connected graph. Since twinless strongly connected components do not share vertices of the same $2$-twinless block, we consider only twinless strongly connected graphs. Let $G=(V,E)$ be a twinless strongly connected graph.  We define a relation $\overset{2t}{\leftrightsquigarrow }$ as follows. For any distinct vertices $x,y\in V$, we write $x \overset{2t}{\leftrightsquigarrow } y$ if for all vertices $w\in V\setminus\lbrace x,y\rbrace$, the vertices $x,y$ are in the same twinless strongly connected component of $G\setminus\lbrace w\rbrace$. By definition, a $2$-twinless block in $G$ is a maximal subset $B^{2t}\subseteq V$ with $|B^{2t}|> 1$ such that for every two vertices $x,y \in B^{2t}$, we have $x \overset{2t}{\leftrightsquigarrow } y$.

The next lemma shows that $2$-twinless blocks share at most one vertex.
\begin{lemma}\label{def:2tbsHasAtMostOneVertex}
	Let $G=(V,E)$ be a twinless strongly connected graph. Let $B^{2t}_1,B^{2t}_2$ be distinct $2$-twinless blocks in $G$. Then $|B^{2t}_1 \cap B^{2t}_2|\leq 1$.
\end{lemma}

\begin{proof}
	Suppose for the sake of contradiction that $B^{2t}_1$ and $B^{2t}_2$ have at least two vertices in common. Clearly, $B^{2t}_1 \cup B^{2t}_2$ is not a $2$-twinless block in $G$.
	Let $x$ and $y$ be vertices belonging to $B^{2t}_1$ and $B^{2t}_2$, respectively, such that $x,y \notin B^{2t}_1 \cap B^{2t}_2$. Let $z$ be any vertex in $V\setminus \lbrace x,y\rbrace$. Since $|B^{2t}_1 \cap B^{2t}_2|>1$, there is a vertex $v$ in $(B^{2t}_1 \cap B^{2t}_2)\setminus\left\lbrace z \right\rbrace $. Note that $x,v$ are in the same twinless strongly connected component of $G\setminus\left\lbrace z\right\rbrace $ since $x,v \in B^{2t}_1$. Moreover, $v$ and $y$ lie in the same twinless strongly connected component of $G\setminus\left\lbrace z\right\rbrace$. By [\cite{SR06}, Lemma $1$] $x$ and $y$ are in the same twinless strongly connected component of $G\setminus\left\lbrace z\right\rbrace$. Therefore, $x,y$ belong to the same $2$-twinless block.   
\end{proof}
The following lemma shows an interesting property of the relation $\overset{2t}{\leftrightsquigarrow }$.
\begin{lemma}\label{def:AllVerticesOfCycleAreIn2twinlessgBlock}
	Let $G=(V,E)$ be a twinless strongly connected graph and let $\left\lbrace v_0,v_1,\ldots,v_l\right\rbrace $ be set of vertices of $G$ such that $v_l \overset{2t}{\leftrightsquigarrow } v_0$ and $v_{i-1} \overset{2t}{\leftrightsquigarrow } v_i$ for $i \in\lbrace 1,2\ldots,l\rbrace$. Then $\left\lbrace v_0,v_1,\ldots,v_l\right\rbrace $ is a subset of a $2$-twinless block in $G$.
\end{lemma}
\begin{proof}	
Assume for the purpose of contradiction that there are two vertices $v_r$ and $v_q$ in $G$ such that $v_r$ and $v_q$ are in distinct $2$-twinless blocks of $G$ and $r,q\in \lbrace 0,1,\ldots,l\rbrace$. Suppose without loss of generality that $r<q$. Then there is a vertex $z\in V\setminus\lbrace v_r,v_q\rbrace$ such that $v_r$ and $v_q$ are in distinct twinless strongly connected components of $G\setminus\left\lbrace z\right\rbrace $. We distinguish two cases.
	\begin{enumerate}
		\item $z \in \lbrace v_{r+1},v_{r+2},\ldots,v_{q-1}\rbrace$. In this case, the vertices $v_{i-1},v_i$ belong to the same twinless strongly connected component of $G\setminus \lbrace z\rbrace$ for each $i\in \lbrace 1,2,\ldots,r\rbrace \cup\lbrace q+1,q+2,\ldots,l\rbrace$. Moreover, the vertices $v_0,v_l$ are in the same twinless strongly connected component of $G\setminus \lbrace z\rbrace$ because  $v_0 \overset{2t}{\leftrightsquigarrow } v_l$. Therefore, the vertices $v_r,v_q$ are in the same twinless strongly connected component of the graph $G\setminus \lbrace z\rbrace$, a contradiction.
			\item $z\notin \lbrace v_{r+1},v_{r+2},\ldots,v_{q-1}\rbrace$. Then, for each $i\in \lbrace r+1,r+2,\ldots,q\rbrace$, the vertices $v_{i-1},v_i$ lie in the same twinless strongly connected component of $G\setminus \lbrace z\rbrace$. Consequently, the vertices $v_r,v_q$ belong to the same twinless strongly connected component of the graph $G\setminus \lbrace z\rbrace$, a contradiction.
	\end{enumerate}
\end{proof}

Let $G=(V,E)$ be a twinless strongly connected graph. We construct the \textit{$2$-twinless block graph} $G^{2t}=(V^{2t},E^{2t})$ of $G$ as follows. For every $2$-twinlesss block $B_{i}$, we add a vertex $v_{i}$ to $V^{2t}$. Moreover, for each vertex $v\in V$, if $v$ belongs to at least two distinct $2$-twinless blocks, we add a vertex $v$ to $V^{2t}$. For any distinct $2$-twinless blocks $B_{i}, B_{j}$ with $B_{i} \cap B_{j}=\left\lbrace v \right\rbrace $, we put two undirected edges $(v_{i},v), (v,v_{j})$ into $E^{2t}$. 
\begin{lemma}\label{def:G2tforest}
 The $2$-twinless block graph of a twinless strongly connected graph is a forest.
\end{lemma}
\begin{proof} This  follows from Lemma \ref{def:AllVerticesOfCycleAreIn2twinlessgBlock} and Lemma \ref{def:2tbsHasAtMostOneVertex}.
\end{proof}

\begin{lemma}\label{def:Relationbetween2sAndtAPs}
	Let $G=(V,E)$ be a twinless strongly connected graph and let $x,y$ be distinct vertices in $G$. Suppose that $v \in V\setminus \lbrace x,y \rbrace$ is not a twinless articulation point. Then $x,y$ are in the same twinless strong connected component of $G\setminus \lbrace v\rbrace$.
\end{lemma}
\begin{proof}
	Immediate from the definition. 
\end{proof}
Algorithm \ref{algo:algor2forall2twinlessblocks} shows an algorithm for computing the $2$-twinless blocks of a twinless strongly connected graph $G=(V,E)$.

\begin{figure}[htbp]
	\begin{myAlgorithm}\label{algo:algor2forall2twinlessblocks}\rm\quad\\[-5ex]
		\begin{tabbing}
			\quad\quad\=\quad\=\quad\=\quad\=\quad\=\quad\=\quad\=\quad\=\quad\=\kill
			\textbf{Input:} A twinless strongly connected graph $G=(V,E)$.\\
			\textbf{Output:} The $2$-twinless blocks of $G$.\\
			{\small 1}\> \textbf{if} $G$ is $2$-vertex-twinless connected \textbf{then}\\
			{\small 2}\>\> Output $V$.\\
			{\small 3}\> \textbf{else}\\
			{\small 4}\>\> Let $S$ be an $n\times n$ matrix.\\
			{\small 5}\>\> Initialize $S$ with $1$s.\\
			{\small 6}\>\> determine the twinless articulation points of $G$.\\
			{\small 7}\>\> \textbf{for} each twinless articulation point $z$ of $G$ \textbf{do} \\
			{\small 8}\>\>\> Identify the twinless strongly connected components of $G\setminus \lbrace z\rbrace$.\\
			{\small 9}\>\>\> \textbf{for} each pair $(v,w) \in (V \setminus\lbrace z\rbrace)\times (V \setminus\lbrace z\rbrace)$ \textbf{do} \\
			{\small 10}\>\>\>\> \textbf{if} $v,w$ in different twinless strongly connected components of $G\setminus \lbrace z\rbrace$ \textbf{then}\\
			{\small 11}\>\>\>\>\> $S[v,w] \leftarrow 0$. \\
			{\small 12}\>\> $E^{b} \leftarrow \emptyset$. \\
			{\small 13}\>\> \textbf{for} each pair $(v,u) \in V\times V $ \textbf{do} \\
			{\small 14}\>\>\> \textbf{if} $S[v,u]=1$ and $S[u,v]=1$ \textbf{then} \\
			{\small 15}\>\>\>\> $E^{b}\leftarrow E^{b} \cup \left\lbrace (v,u)\right\rbrace  $. \\
			{\small 16}\>\> calculate the blocks of size $\geq 2$ of $G^{b}=(V,E^{b})$ and output them. 
		\end{tabbing}
	\end{myAlgorithm}
\end{figure}

The correctness of Algorithm \ref{algo:algor2forall2twinlessblocks} follows from the following lemma. 
\begin{lemma}
	A vertex subset $B\subseteq V$ is a $2$-twinless block of $G$ if and only if $B$ is a block of the undirected graph $G^{b}=(V,E^{b})$ which is constructed in lines $12$--$15$ of Algorithm \ref{algo:algor2forall2twinlessblocks}  
\end{lemma}
\begin{proof}
	It follows from Lemma \ref{def:AllVerticesOfCycleAreIn2twinlessgBlock} and Lemma \ref{def:Relationbetween2sAndtAPs}.
\end{proof}
\begin{theorem}
	Algorithm \ref{algo:algor2forall2twinlessblocks} runs in $O(n^{3})$ time.
\end{theorem}
\begin{proof}
Georgiadis and Kosinas \cite{GeorgiadisKosinas20} showed that the twinless articulation points can be computed in linear time.
	The initialization of matrix $S$ takes $O(n^{2})$ time.
	The number of iterations of the for-loop in lines $7$--$11$ is at most $n$ because the number of twinless articulation points
	is at most $n$. Consequently, lines $7$--$11$ require $O(n^{3})$. Furthermore,  the blocks of an undirected graph can be found in linear time \cite{T72,JS13}. 
\end{proof}

The following lemma shows an important property of $G^{b}$.
\begin{lemma}\label{def:Gbischordal}
The graph $G^{b}$ which is constructed in lines $12$--$15$ of Algorithm \ref{algo:algor2forall2twinlessblocks} is chordal. 
\end{lemma}
\begin{proof}
	It follows from Lemma \ref{def:AllVerticesOfCycleAreIn2twinlessgBlock}.
\end{proof}
By Lemma \ref{def:Gbischordal}, one can calculate the maximal cliques of  $G^{b}$ instead of blocks. The maximal cliques of a chordal graph can be calculated in linear time \cite{G72,RT75}.
\section{An improved version of Algorithm \ref{algo:algor2forall2twinlessblocks}}
In this section we present an improved version of Algorithm \ref{algo:algor2forall2twinlessblocks}. 

The following lemma shows a connection between $2$-twinless blocks and $2$-strong blocks.
\begin{lemma}\label{def:Each2tIsIn2sb}
	Let $G=(V,E)$ be a twinless strongly connected graph. Suppose that $B_{t}$ is a $2$-twinless block in $G$. Then $B_{t}$ is a subset of a $2$-strong block in $G$. 
\end{lemma}
\begin{proof}
	Let $v$ and $w $ be distinct vertices in $B_{t}$, and let $x\in V\setminus\lbrace v,w\rbrace$. By definition, the vertices $v,w$ belong to the same twinless strongly connected component $C$ of $G\setminus\lbrace x\rbrace$. Since $C$ is a subset of a strongly connected component of $G$, the vertices  $v,w$ also lie in the strongly connected component of $G\setminus\lbrace x\rbrace$. Consequently, $v,w$ are in the same $2$-strong block in $G$.
\end{proof}

Algorithm \ref{algo:improvedalgor2forall2twinlessblocks} describes this improved version which is based on Lemma \ref{def:Each2tIsIn2sb} and Lemma \ref{def:AllVerticesOfCycleAreIn2twinlessgBlock}.
\begin{figure}[h]
	\begin{myAlgorithm}\label{algo:improvedalgor2forall2twinlessblocks}\rm\quad\\[-5ex]
		\begin{tabbing}
			\quad\quad\=\quad\=\quad\=\quad\=\quad\=\quad\=\quad\=\quad\=\quad\=\kill
			\textbf{Input:} A twinless strongly connected graph $G=(V,E)$.\\
			\textbf{Output:} The $2$-twinless blocks of $G$.\\
			{\small 1}\> \textbf{if} $G$ is $2$-vertex-twinless connected \textbf{then}\\
			{\small 2}\>\> Output $V$.\\
			{\small 3}\> \textbf{else}\\
			{\small 4}\>\> find the $2$-strong blocks of $G$\\
			{\small 5}\>\> Let $S$ be an $n\times n$ matrix.\\
			{\small 6}\>\> Initialize $S$ with $0$.\\
			{\small 7}\>\> $A\leftarrow \emptyset $.\\
			{\small 8}\>\> \textbf{for} each $2$-strong block $B$ of $G$ \textbf{do} \\
			{\small 9}\>\>\>\textbf{for} each pair of vertices $v,w \in B$ \textbf{do} \\
			{\small 10}\>\>\>\> $S[v,w] \leftarrow 1$ \\
			{\small 11}\>\>\>\> $S[w,v] \leftarrow 1$ \\
			{\small 12}\>\>\>\textbf{for} each vertex $v \in B$ \textbf{do} \\
			{\small 13}\>\>\>\>\textbf{if} $v \notin A$ \textbf{then}\\
			{\small 14}\>\>\>\>\>\> add $v$ to $A$\\
			{\small 15}\>\> determine the twinless articulation points of $G$.\\
			{\small 16}\>\> \textbf{for} each twinless articulation point $z$ of $G$ \textbf{do} \\	
			{\small 17}\>\>\>\> Identify the twinless strongly connected components of $G\setminus \lbrace z\rbrace$.\\
			{\small 18}\>\>\>\> \textbf{for} each pair $(v,w) \in (A \setminus\lbrace z\rbrace)\times (A \setminus\lbrace z\rbrace)$ \textbf{do} \\
			{\small 19}\>\>\>\>\> \textbf{if} $v,w$ in different twinless strongly connected components of $G\setminus \lbrace z\rbrace$ \textbf{then}\\
			{\small 20}\>\>\>\>\>\> $S[v,w] \leftarrow 0$. \\
			{\small 21}\>\> $E^{b} \leftarrow \emptyset$. \\
			{\small 22}\>\> \textbf{for} each pair $(v,u) \in A\times A $ \textbf{do} \\
			{\small 23}\>\>\> \textbf{if} $S[v,u]=1$ and $S[u,v]=1$ \textbf{then} \\
			{\small 24}\>\>\>\> $E^{b}\leftarrow E^{b} \cup \left\lbrace (v,u)\right\rbrace  $. \\
			{\small 25}\>\> calculate the blocks of size $\geq 2$ of $G^{b}=(A,E^{b})$ and output them. 
		\end{tabbing}
	\end{myAlgorithm}
\end{figure}

\begin{theorem}
	The running time of Algorithm \ref{algo:improvedalgor2forall2twinlessblocks} is $O(t(s^{2}+m)+n^{2})$, where $s=|A|$ and $t$ is the number of twinless articulation points of $G$.
\end{theorem}
\begin{proof}
	The $2$-strong blocks of a directed graph can be computed in linear time \cite{GILP14VertexConnectivity}. Furthermore, the twinless articulation points of a directed graph can be identified in linear time using the algorithm of Georgiadis and Kosinas \cite{GeorgiadisKosinas20}. Since the number of iterations of the for-loop in lines $16$--$20$ is at most $t$, lines  $16$--$20$ take $O(t(s^{2}+m))$ time.
\end{proof}

	Let $G=(V,E)$ be a twinless strongly connected graph. If the refine operation defined in \cite{LGILP15} is used to refine the $2$-strong blocks of $G$ for all twinless articulation points, then the $2$-twinless blocks of a directed graph $G=(V,E)$ can be computed in $O(tm)$  time, where $t$ is the number of twinless articulation points of $G$.

We leave as open problem whether the $2$-twinless blocks of a directed graph can be calculated in linear time.
\section*{Acknowledgements.}
The author would like to thank the anonymous reviewers for their helpful comments and suggestions.


\begin{thebibliography}{4}
		\bibitem {AHLT99} S. Alstrup, D. Harel, P.W. Lauridsen, M. Thorup, Dominators in linear time,
		SIAM J. Comput. $28$($6$) ($1999$) $2117$--$2132$. 
		\bibitem {BGKRTW00} A.L. Buchsbaum, L. Georgiadis, H. Kaplan, A. Rogers, R.E. Tarjan, J.R. Westbrook, Linear-time algorithms for dominators and
		other path-evaluation problems, SIAM J. Comput. $38$($4$) ($2008$) $1533$--$1573$.
	\bibitem {FILOS12} D. Firmani, G.F. Italiano, L. Laura, A. Orlandi, F. Santaroni, Computing strong articulation points and strong bridges in large scale graphs, SEA, LNCS $7276$, ($2012$) $195$--$207$. 
	\bibitem{FILOS16} D. Firmani, L. Georgiadis, G. F. Italiano, L. Laura, F. Santaroni,
	Strong Articulation Points and Strong Bridges in Large Scale Graphs. Algorithmica $74(3): 1123$--$1147 (2016)$
	\bibitem {G72} F. Gavril, Algorithms for Minimum Coloring, Maximum Clique, Minimum Covering by Cliques, and Maximum Independent Set of a Chordal Graph. SIAM J. Comput. $1(2)$ ($1972$) $180$--$187$.
	
	\bibitem{LTW2006}
	Loukas Georgiadis, Robert Endre Tarjan, Renato Fonseca F. Werneck,
	Finding Dominators in Practice. J. Graph Algorithms Appl. $10(1): 69$--$94 (2006)$
	\bibitem {G10} L. Georgiadis, Testing $2$-vertex connectivity and computing pairs of vertex-disjoint s-t paths in digraphs, In Proc. $37$th ICALP, Part I, LNCS $6198$ ($2010$) $738$--$749$.
    \bibitem{GeorgiadisKosinas20}L. Georgiadis, E. Kosinas,
    Linear-Time Algorithms for Computing Twinless Strong Articulation Points and Related Problems, ISAAC $2020: 38:1$--$38:16$
	\bibitem {GILP14SODA} L. Georgiadis, G.F. Italiano, L. Laura, N. Parotsidis, $2$-Edge Connectivity in Directed Graphs, SODA ($2015$) $1988$--$2005$.
		\bibitem{LILPACM2016} L. Georgiadis, G. F. Italiano, L. Laura, N. Parotsidis,
	$2$-Edge Connectivity in Directed Graphs. ACM Trans. Algorithms $13(1): 9:1$--$9:24 (2016)$
	\bibitem {GILP14VertexConnectivity} L. Georgiadis, G.F. Italiano, L. Laura, N. Parotsidis, $2$-Vertex Connectivity in Directed Graphs, ICALP $(1)2015: 605$--$616$
		\bibitem{GILP2018}L. Georgiadis, G. F. Italiano, L. Laura, N. Parotsidis,
	$2$-vertex connectivity in directed graphs. Inf. Comput. $261: 248$--$264 (2018)$
	\bibitem{GIKPALENEX2018} L. Georgiadis, G. F. Italiano, A. Karanasiou, N. Parotsidis, N. Paudel,
	Computing $2$-Connected Components and Maximal $2$-Connected Subgraphs in Directed Graphs: An Experimental Study. ALENEX $2018: 169$--$183$
	\bibitem{HKL15} M. Henzinger, S. Krinninger, V. Loitzenbauer,
	Finding $2$-Edge and $2$-Vertex Strongly Connected Components in Quadratic Time. ICALP $(1) 2015:$ $713$--$724$
	\bibitem {ILS12} G.F. Italiano, L. Laura, F. Santaroni,
	Finding strong bridges and strong articulation points in linear time, Theoretical Computer Science $447$ ($2012$) $74$--$84$.
	\bibitem{Jaberi16} R. Jaberi,
	On computing the 2-vertex-connected components of directed graphs. Discrete Applied Mathematics 204: $(2016)164$--$172$ 
	\bibitem {J14} R. Jaberi, On Computing the $2$-vertex-connected components of directed graphs, (January, $2014$) CoRR abs/$1401.6000$ .
	\bibitem{Jaberi15} R. Jaberi,
	Computing the 2-blocks of directed graphs. RAIRO - Theor. Inf. and Applic. $49(2) (2015) 93$--$119 $
	\bibitem{Jaberi19} R. Jaberi, Twinless articulation points and some related problems, $2019$, arxiv, abs$/1912.11799$
	\bibitem{LGILP15} W. D. Luigi, L. Georgiadis, G. F. Italiano, L. Laura, N. Parotsidis,
	2-Connectivity in Directed Graphs, An Experimental Study. ALENEX. ($2015$) $173$--$187$.
	\bibitem {LT79} T. Lengauer, R.E. Tarjan, A fast algorithm for finding dominators in a flowgraph. ACM Trans. Program. Lang. Syst. $1(1)$ ($1979$) $121$--$141$.
	\bibitem {SR06} S. Raghavan, Twinless Strongly Connected Components, Perspectives in Operations Research, ($2006$) $285$--$304$.
		\bibitem {RT75} D.J. Rose, R.E. Tarjan, Algorithmic Aspects of Vertex Elimination. STOC $(1975)$ $245$--$254$.
	\bibitem {JS13} J. Schmidt, A Simple Test on 2-Vertex- and 2-Edge-Connectivity, Information Processing Letters, $113$ ($7$) ($2013$) $241$–-$244$
	\bibitem {T72} R.E. Tarjan, Depth-first search and linear graph algorithms, SIAM J. Comput. $1$($2$) ($1972$) $146$--$160$	
\end{thebibliography}
\end{document}